\documentclass{fundam}

\usepackage{amsmath}
\usepackage{amssymb}
\usepackage{multirow}

\usepackage{comment}
\usepackage{enumerate}

\usepackage{graphicx}
\graphicspath{ {./Figures/} }
\usepackage{adjustbox}

\usepackage{url}
\usepackage[utf8]{inputenc}
\usepackage[T1]{fontenc}
\DeclareMathOperator{\Reg}{\mathsf{Reg}}

\usepackage{hyperref}

\begin{document}

 \setcounter{page}{299}
 \publyear{24}
 \papernumber{2184}
  \volume{191}
  \issue{3-4}

\finalVersionForARXIV


\title{Strong Regulatory Graphs}

\author{Patric Gustafsson\thanks{This work was initiated as part of Patric Gustafsson's master thesis project at \AA
                   bo Akademi University (\AA A) in 2019. Patric is no longer affiliated with \AA A.}
 \\
Department of Information Technologies\\
 \AA bo Akademi University, Finland\\
patricjgustafsson08@gmail.com
\and Ion Petre\thanks{Address for correspondence: Department of Mathematics and Statistics, University of Turku, Finland}
\\
Department of Mathematics and Statistics\\
University of Turku, Finland\\
National Institute of Research and \\
Development for Biological Sciences, Romania\\
ion.petre@utu.fi
}

\maketitle

\runninghead{P. Gustafsson and I. Petre}{Strong Regulatory Graphs}

\begin{abstract}
Logical modeling is a powerful tool in biology, offering a system-level understanding of the complex interactions that govern biological processes. A gap that hinders the scalability of logical models is the need to specify the update function of every vertex in the network depending on the status of its predecessors. To address this, we introduce in this paper the concept of strong regulation, where a vertex is only updated to active/inactive if all its predecessors agree in their influences; otherwise, it is set to ambiguous. We explore the interplay between active, inactive, and ambiguous influences in a network. We discuss the existence of phenotype attractors in such networks, where the status of some of the variables is fixed to active/inactive, while the others can have an arbitrary status, including ambiguous.
\end{abstract}

\begin{keywords}
Biomodeling, interaction networks, regulatory graphs, Boolean networks, phenotype attractors.
\end{keywords}

\section{Introduction}
\label{sec:intro}

Logical modeling is a well established mathematical approach to biology, offering a simple, intuitive way to understand complex systems. Biological networks can be represented as graphs and experimental data integrated in them to yield a system-level view of key regulators, dynamical patterns, and response to various changes. There are many successful applications of logical modeling such as the identification of drug targets and treatments \cite{Bloomingdale2018}, modeling of regulatory networks \cite{CACACE2020205}, of signaling pathways \cite{10.7554/eLife.72626}, of cell growth and apoptosis \cite{10.1371/journal.pcbi.1006402}, applications in immunology \cite{10.3389/fphys.2020.590479}, and many others. A logical model is discrete (often Boolean), with its dynamics defined as the result of the interplay of the influences among its variables. Such influences are modeled as directed edges between the variables, sometimes with the annotation of the activation/inhibition nature of the interaction, see \cite{Wang_2012}  for a review on Boolean modeling in systems biology.

Constructing large-scale models has been made possible in recent years by the availability of data respositories such as KEGG \cite{10.1093/nar/gkv1070}, OmniPath \cite{article-omnipath}, InnateDB \cite{article-innate}, SIGNOR \cite{article-signor} and DrugBank \cite{article-drugbank}. Using such resources has made it possible to build models consisting of thousands of variables and interactions \cite{10.1093/bib/bbab490, NetControlGenAlg}. Integrating such large scale data into logical or Boolean models is problematic because of the need to identify the update functions for each variable of the model. These  functions specify the precise conditions under which a variable changes its status, e.g., between active and inactive, depending on the status of its predecessors/regulators in the network. The difficulty is that the data is almost never enough to identify the precise nature of these  functions.
Instead, the modeler often postulates the type of update functions in the network (propositional formulas (\cite{GTZANUDO20181}), threshold conditions (\cite{zanudo2011boolean}), multi-valued functions (\cite{naldi_decision_2007})) The specifics of how to choose each update function are based on previous literature and to some extent, to the choice of the modeler, subject to experimentation with various options.

The approach we take in this paper is motivated by the concept of controlling a network: choose a user intervention in the network (e.g., through fixing some variables to some constants), so that the model converges to a certain desired attractor. Network controllability has found interesting applications in biology, e.g., in drug repurposing \cite{10.1093/bib/bbab490, NetControlGenAlg}. However, tracking the effect of the user interventions throughout the network is difficult when the variables of the model are under conflicting influences, with some regulators pushing for their activation, while others for their inhibition. The computational complexity of this problem has been investigated in \cite{StructuralTargetControl2018, ControlRSystems2020}.
We propose in this paper a notion of \emph{strong regulatory graphs} where the update of a node's status is  determined to be active/inactive \emph{only if its predecessors concur in their influences}. Otherwise it is set to be \emph{ambiguous}, meaning that it could be both active and inactive, depending on the precise (numerical) setup of its predecessors and of their influences, which is in practice very difficult to determine. This leads to an intricate interplay between ambiguous, negative and positive influences in our framework.
Interesting questions about the spread of ambiguity in the network can be asked, e.g., in terms of the existence of phenotype attractors. We introduce in this paper a simple mathematical formalization of the concept of \emph{strong regulation}. We discuss the phenotype attractor problem for strong regulatory graphs: whether attractors exist in which some of the variables are constant (they follow a given active/inactive phenotype), while the others can be arbitrary, even ambiguous.

\section{Strong regulatory graphs (SRG)}
The model proposed here builds on the regulatory graphs structure \cite{thomas1991regulatory, thomas1995dynamical, GIN-sim, naldi_decision_2007} and adds to them the concept of strong regulation.
Choosing how to update the status of a node is obvious if its active regulators concur in their influences: if they all exert an activation regulation on the node, then it gets \emph{activated} (denoted as $1$), and if all exert an inhibition regulation on it, then it gets \emph{inhibited} (denoted as $-1$). A node under conflicting regulation, with some active regulators trying to activate it and others to inhibit it, is less clear how to update. The typical approach is to postulate an update function, Boolean or multi-valued, that sets its status to active or inactive, based on a specific schema of the status of its predecessors. The main difficulty of these approaches is that the experimental data and the level of detail in the model are rarely sufficient to specify such detailed update functions. Our proposal instead is that under such conflicting the status of the node should simply be recognized as being \emph{ambiguous} (denoted as $0$). This means that the node may potentially be active or inactive, but the level of detail in the model is not enough to specify it with confidence either way. The ambiguity of a node may well cascade down through the system, in that the regulation it exerts on its successor in the regulatory graph may be active or not. Still, an ambiguous regulation may well be resolved by another \emph{active} regulation on the same node.
This opens the possibility to reason about the dynamic interplay between activation, inhibition, and ambiguity, where the semantics of activation and inhibition is in its \emph{strongest} possible form to mean that the node has been \emph{unanimously} activated/inhibited. We define the structure of the strong regulatory graphs in Definition~\ref{def_SRG} and their dynamics in Definition~\ref{d_SRG}.

For a graph $G=(V,E)$, we consider the set of vertices to be ordered. In this way, when we write a vector $x\in\{-1,0,1\}^{|V|}$, $x_v$ will denote the component of $x$ corresponding to $v\in V$, without risk of ambiguity.

\begin{definition}
\label{def_SRG}
    A strong regulatory graph (SRG) is an edge-labeled graph \(G=(V, E)\) where $V=\{v_1,v_2,\ldots,v_n\}$ is the finite set of \emph{vertices} and $E\subseteq V\times V$ is the set of directed edges. The labeling of the edges is done through the partition $\{E_+,E_-\}$ of $E$. We say that $(u,v)\in E_+$ is an \emph{activating edge} and that $(u,v)\in E_-$ is an \emph{inhibiting edge}.

    A state of the SRG is $x\in\{-1,0,1\}^{|V|}$ whose intended meaning is that $x_v=-1$ if $v$ is inactive in state $x$, $x_v=1$ if $v$ is active in state $x$, and $x_v=0$ if $v$ is ambiguous in state $x$.
    For an edge $(u,v)\in E_+$, we say that $u$ is an \emph{activator} of $v$ in state $x$ if $x_u=1$. We say that it is a \emph{potential activator} of $v$ in $x$ if $x_u=0$.
    Similarly, for an edge $(u,v)\in E_-$, we say that $u$ is an \emph{inhibitor} of $v$ in state $x$ if $x_u=1$. We say that it is a \emph{potential inhibitor} of $v$ in $x$ if $x_u=0$.
\end{definition}

The dynamics of a strong regulatory graph $G=(V,E)$ is defined through a state-transition system. The state is given by the $-1$/$0$/$1$ activation status of all nodes in the graph. The set of states is thus $\{-1,0,1\}^{|V|}$. Given a state, a transition will indicate the change in the activation status of the nodes in the graph.
The states can be updated in a synchronous or in an asynchronous way.
For simplicity, we only discuss in this paper the synchronous activation update and the dynamics of strong regulatory graph in terms of deterministic state-transition systems. The asynchronous update strategy leads to a similar discussion.

In defining the update rule for the status of a node, we follow our proposal for the concept of \emph{strong regulation}, where a node is activated/inhibited if all its regulators concur in their influences on the node. In case of conflicting influences, the status of the node is set to `ambiguous'. The influence of a regulator whose status is ambiguous is quite subtle, depending on the status of the other regulators and of the node itself. We discuss the intuition of the update rule below and then give it a formal definition.

\medskip
If there are no potentially active inhibitors and no potentially active activators, then the vertex should preserve its activation status. A node $v$ should be set to `active' in two situations:
\begin{itemize}
\item $v$ was active and none of its inhibitors were potentially active, i.e., active or ambiguous (in this case, the node remains active, absent any potential or active inhibition regulation), or
\item at least one activator of $v$ was active and none of its inhibitors were potentially active (in this case, the regulation on the node is un-ambiguous towards activation).
\end{itemize}
Similarly, a node $v$ should be set to `inactive' after an application of the update rule in two situations:
\begin{itemize}
\item $v$ was inactive and none of its activators were potentially active, i.e., active or ambiguous, or
\item at least one inhibitor of $v$ was active and none of its activators were potentially active.
\end{itemize}
If the node is under conflicting influences, its status should be set to `ambiguous'. Here are four cases when this should happen:
\begin{itemize}
\item at least one activator and at least one inhibitor of $v$ were potentially active; this reflects a node under potentially conflicting regulation;
\item $v$ was active and one of its inhibitors was ambiguous (regardless of the status of its activators); this reflects the ambiguity of whether there is an inhibitory influence on $v$ or not, which may either leave the node active, or switch to inactive;
\item $v$ was inactive and one of its activators was ambiguous (regardless of the status of its inhibitors); this reflects the ambiguity of whether there is an activating influence on $v$ or not, which may either leave the node inactive, or switch it to active;
\item $v$ was ambiguous and none of its regulators is active (i.e., they are either inactive or ambiguous); this reflects the situation where the status of $v$ cannot be clarified because the status of its regulators is either inactive or ambiguous.
\end{itemize}

For a vertex $v\in V$ and a state $x\in\{-1,0,1\}^{|V|}$, we define the set $\Reg_+(x, v)$ of activators of $v$ in state $x$ and its set $\Reg_-(x, v)$ of inhibitors as follows:
\begin{align*}
    \Reg_+(x, v)&=\{x_u \mid (u,v)\in E_+, x_u\in\{0,1\}\}; \\
    \Reg_-(x, v)&=\{x_u\mid (u,v)\in E_-, x_u\in\{0,1\}\}.
\end{align*}
The set $\Reg_+(x, v)$ contains $1$ if $v$ has at least one active activator in state $x$ and it contains $0$ if $v$ has at least one ambiguous activator in state $x$.
Similarly, the set $\Reg_-(x, v)$ contains $1$ if $v$ has at least one active inhibitor in state $x$ and it contains $0$ if $v$ has at least one ambiguous inhibitor in state $x$.

We also define the reflexive extension of $\Reg$ to include the status of the vertex itself. The reflexive extension captures a sort of ``inertia'' of the dynamics (defined below) where a node may preserve its status in the absence of active regulators.
\begin{align*}
    \overline{\Reg}_+(x, v)&=\Reg_+(x, v)\cup \{x_v\mid x_v\in\{0,1\}\}; \\
    \overline{\Reg}_-(x, v)&=\Reg_-(x, v)\cup \{-x_v\mid x_v\in\{-1, 0\}\}.
\end{align*}
Note that an inactive node $v$, i.e. $x_v=-1$, contributes $1$ to $\overline{\Reg}_-(x, v)$, consistent with our suggestion of ``inertia'': absent any active regulation on node $v$, it will remain inactive.

\begin{definition}\label{d_SRG}
Let $G=(V,E)$ be a strong regulatory graph. Its dynamics is given by the map $f_G:\{-1,0,1\}^{|V|}\rightarrow \{-1,0,1\}^{|V|}$ defined as
\[
f_G(x)_v=
\begin{cases}
1, \text{ if } 1\in\overline{\Reg}_+(x,v) \text{ and } \Reg_-(x,v)=\emptyset; \\
-1, \text{ if } 1\in\overline{\Reg}_-(x,v) \text{ and } \Reg_+(x,v)=\emptyset; \\
0, \text{ otherwise,}
\end{cases}
\]
for all $v\in V$. The dynamics of a strong regulatory graph can also be seen as a state-transition system with the set of states being $\{-1,0,1\}^{|V|}$ and the transitions defined through the map $f_G$.
\end{definition}

\begin{figure}[htb]
\begin{center}
\begin{tabular}{ccc}
\includegraphics[width=3.72cm]{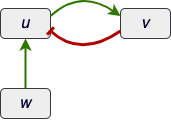}
&
\hspace*{1cm}
&
{
\includegraphics[width=3.72cm]{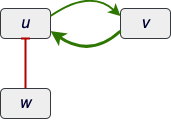}
}
\\
(a) & & (b)
\end{tabular}\vspace*{-5mm}
\end{center}
\caption{Two simple regulatory graphs \cite{naldi_decision_2007}.
The vertices are shown with rectangles, the activation edges with pointed arrows and the inhibition edges with blunt
arrows.}\label{fig-naldi}\vspace*{-4mm}
\end{figure}

\begin{example}\label{ex_two_srg}
We discuss two regulatory graphs introduced in \cite{naldi_decision_2007} and shown in Figure \ref{fig-naldi}. The set of vertices is $V=\{u, v, w\}$,
the activation edges are shown with pointed arrows and the inhibition edges with blunt arrows.

\medskip
For the graph in Figure \ref{fig-naldi}(a), consider the state $x=(-1,1,1)$. Then:
\begin{itemize}
    \item $\Reg_+(x,u)=\{1\}$, since $(w,u)$ is an activation edge and ${x}_w=1$;
    \item $\Reg_-(x,u)=\{1\}$, since $(v,u)$ is an inhibition edge and ${x}_{v}=1$;
    \item $\overline{\Reg}_+(x,u)=\{1\}$ and $\overline{\Reg}_-(x,u)=\{1\}$, since ${x}_{u}=-1$;
    \item $\Reg_+(x,v)=\emptyset$, since the only incoming activation edge into $v$ is $(u,v)$ but ${x}_u=-1$;
    \item $\Reg_-(x,v)=\emptyset$, since there are no incoming inhibition edges into $v$;
    \item $\overline{\Reg}_+(x,v)=\{1\}$ and $\overline{\Reg}_-(x,v)=\emptyset$, since $x_v=1$;
    \item $\Reg_+(x,w)=\Reg_-(x,w)=\emptyset$, since there are no incoming edges into $w$;
    \item $\overline{\Reg}_+(x,$ $w)=\{1\}$ and $\overline{\Reg}_-(x,w)=\emptyset$, since $x_w=1$.
\end{itemize}
This means that $x\rightarrow (0,1,1)$. This corresponds well to our intuition for the meaning of the strong regulatory graphs: $u$ is under conflicting regulation and it becomes set to \emph{ambiguous}, $v$ and $w$ are under no active regulation and they preserve their status. Denote $y=(0,1,1)$. Then:
\begin{itemize}
    \item $\Reg_+(y,u)=\{1\}$, since $(w,u)$ is an activation edge and $y_w=1$;
    \item $\Reg_-(y,u)=\{1\}$, since $(v,u)$ is an inhibition edge and $y_v=1$;
    \item $\overline{\Reg}_+(y,u)=\{0,1\}$ and $\overline{\Reg}_-(y,u)=\{0,1\}$, since $y_u=0$;
    \item $\Reg_+(y,v)=\{0\}$, since $(u,v)$ is an activation edge and $y_u=0$;
    \item $\Reg_-(y,v)=\emptyset$, since there are no incoming inhibition edges into $v$;
    \item $\overline{\Reg}_+(y,v)=\{0,1\}$ and $\overline{\Reg}_-(y,v)=\emptyset$, since $y_v=1$;
    \item $\Reg_+(y,w)=\Reg_-(y,w)=\emptyset$, since there are no incoming edges into $w$;
    \item $\overline{\Reg}_+(y,$ $w)=\{1\}$ and $\overline{\Reg}_-(y,w)=\emptyset$, since $y_w=1$.
\end{itemize}
Consequently, $y\rightarrow y$. The intuition for the update on $u$ and $w$ remains the same. For $v$, there is a potential activation regulation coming from the ambiguous status of its activator $u$, but since $v$ is active in $y$, it will remain so in the next state.

\medskip
For the graph in Figure \ref{fig-naldi}(b), the following are valid state transitions:
\begin{itemize}
    \item $(-1,1,-1)\rightarrow(1,1,-1)\rightarrow(1,1,-1)$,
    \item $(1,-1,-1)\rightarrow(1,1,-1)$,
    \item $(-1,-1,1)\rightarrow(-1,-1,1)$ and
    \item $(1,-1,1)\rightarrow(-1,1,1)\rightarrow(0,1,1)\rightarrow(0,1,1)$.
\end{itemize}

The state transition graphs for these two examples are in Figure \ref{fig_two_sts}.
\end{example}

\begin{figure}[!ht]
\vspace*{-5mm}
\centering
\begin{tabular}{c}
\adjustbox{
clip=true, trim=1cm 1cm 1cm 1cm 
}
{\hspace*{-8mm}
\includegraphics[width=1.18\textwidth]{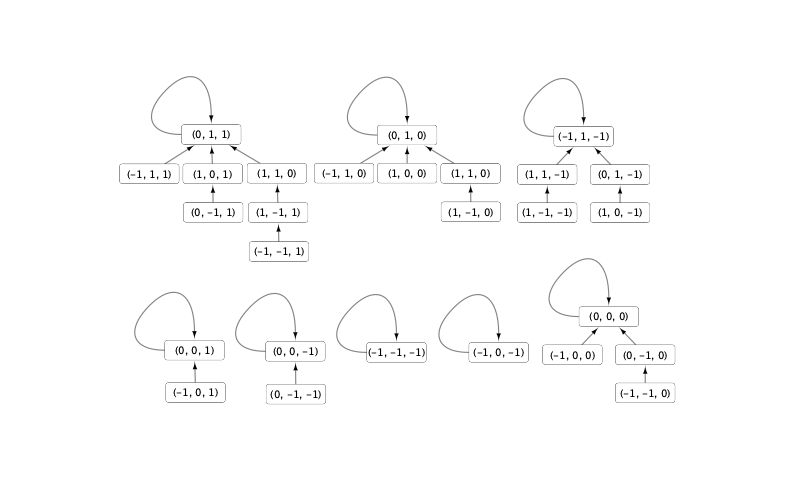}
}
\vspace*{-8mm}\\
(a)
\\[-15pt]
\adjustbox{
clip=true, trim=1.5cm 2.5cm 1.5cm 1.5cm 
}
{\hspace*{-20mm}
\includegraphics[width=1.36\textwidth]{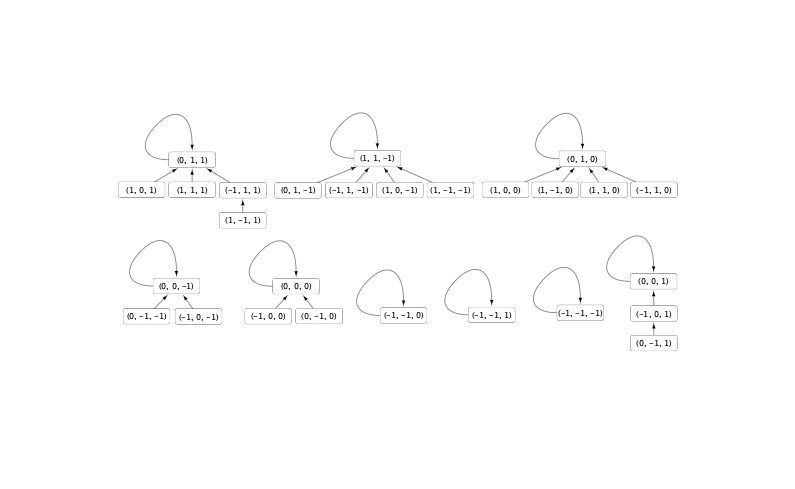}
}
\vspace*{-10mm}\\
(b)
\end{tabular}
\caption{The state transition graphs of the regulatory graphs in Example \ref{ex_two_srg} and Figure \ref{fig-naldi}. }
\label{fig_two_sts}\vspace*{-3mm}
\end{figure}

The following observation shows that a vertex preserves its activation status, in the absence of potentially active predecessors.

\begin{lemma}\label{l_nopred}
Let $G=(V,E)$ be a strong regulatory graph and $v\in V$. For any state $x$ with $\Reg_+(x,v)=\Reg_-(x,v)=\emptyset$, we have $f_G(x)_v=x_v$.
\end{lemma}

\begin{proof}
If $x_v=-1$, then $\overline\Reg_-(x,v)=\{1\}$ and so, $f_G(x)_v=-1$.
If $x_v=1$, then $\overline\Reg_+(x,v)=\{1\}$ and so, $f_G(x)_v=1$.
If $x_v=0$, then $\overline\Reg_-(x,v)=\overline\Reg_+(x,v)=\{0\}$ and so, $f_G(x)_v=0$.
\end{proof}

We clarify now the conditions under which a vertex $v$ is updated to an `ambiguous' state. Following Lemma \ref{l_nopred}, we only focus on the case when $v$ is under some (potential) regulators, i.e., when either $\Reg_+(x,v)\ne\emptyset$, or $\Reg_-(x,v)\ne\emptyset$, as otherwise $v$ maintains its status unchanged. The theorem shows that our definition covers exactly the intuition offered just before Definition \ref{def_SRG}.

\begin{theorem}\label{t_SRG}
Let $G=(V,E)$ be an SRG and $v\in V$ in state $x$, with $\Reg_+(x,v)\cup\Reg_-(x,v)\ne\emptyset$. We have $f_G(x)_v=0$ if and only if
\begin{enumerate}[(i)]
    \item $\Reg_+(x,v)\ne\emptyset$ and $\Reg_-(x,v)\ne\emptyset$, or
    \item $\overline{\Reg}_+(x,v)=\{0\}$ and $\Reg_-(x,v)=\emptyset$, or
    \item $\overline{\Reg}_-(x,v)=\{0\}$ and $\Reg_+(x,v)=\emptyset$.
\end{enumerate}
\end{theorem}

\begin{proof}

\vspace*{-15mm}
\begin{align}
\notag &f_G(x)_v=0 \Leftrightarrow f_G(x)_v\ne 1 \land f_G(x)_v\ne -1\\
\notag \Leftrightarrow \ & (1\not\in\overline{\Reg}_+(x,v) \lor \Reg_-(x,v)\ne\emptyset) \land
(1\not\in\overline{\Reg}_-(x,v)\lor \Reg_+(x,v)\ne\emptyset)\\
\label{align1}
\Leftrightarrow
\ &(1\not\in\overline{\Reg}_+(x,v) \land 1\not\in\overline{\Reg}_-(x,v)) \\
\label{align2}
\ & \lor (1\not\in\overline{\Reg}_+(x,v) \land \Reg_+(x,v)\ne\emptyset) \\
\label{align3}
\ & \lor(\Reg_-(x,v)\ne\emptyset \land 1\not\in\overline{\Reg}_-(x,v))\\
\label{align4}
\ & \lor(\Reg_-(x,v)\ne\emptyset \land \Reg_+(x,v)\ne\emptyset)
\end{align}

Because of \eqref{align4} and the hypothesis that $\Reg_+(x,v)\cup \Reg_-(x,v)\ne\emptyset$, we can consider in \eqref{align1}, \eqref{align2} and in \eqref{align3} that one of the sets $\Reg_+(x,v)$ and $\Reg_-(x,v)$ is non-empty, while the other is empty.

\medskip
In \eqref{align2}, we observe that $1\not\in\overline{\Reg}_+(x,v)$ is equivalent to $x_v\ne 1 \land 1\not\in{\Reg}_+(x,v)$, which together with $\Reg_+(x,v)\ne\emptyset$ is equivalent with $x_v\ne 1 \land \Reg_+(x,v)=\{0\}$. As we noted above, we can add to the conjunction also the term $\Reg_-(x,v)=\emptyset$ since $\Reg_+(x,v)\ne\emptyset$. This gives us the clause $x_v\ne 1 \land \Reg_+(x,v)=\{0\} \land \Reg_-(x,v)=\emptyset$. Under the hypothesis of the theorem that $\Reg_+(x,v)\cup\Reg_-(x,v)\ne\emptyset$, this clause is equivalent with
\begin{align}
\label{2prime}\tag{2'}\overline{\Reg}_+(x,v)=\{0\}\land \Reg_-(x,v)=\emptyset.
\end{align}

Using a symmetric argument, we can conclude that \eqref{align3} can be replaced in the disjunction with
\begin{align}
\label{3prime}\tag{3'}\overline{\Reg}_-(x,v)=\{0\}\land\Reg_+(x,v)=\emptyset.
\end{align}

Clause \eqref{align1} is equivalent with $x_v=0 \land 1\not\in\Reg_+(x,v) \land 1\not\in\Reg_-(x,v)$. Since $\Reg_+(x,v)$, $\Reg_-(x,v)\subseteq\{0,1\}$ and, as noted above, we can assume that one is empty, while the other is not, it follows that clause \eqref{align1} can be replaced with
\begin{align*}
&\left(x_v=0 \land \Reg_+(x,v)=\{0\} \land \Reg_-(x,v)=\emptyset\right)  \lor \\
&\left(x_v=0 \land \Reg_-(x,v)=\{0\} \land \Reg_+(x,v)=\emptyset\right).
\end{align*}
This is equivalent with
\begin{equation}
\tag{1'}
\begin{tabular}{c}
$\left(x_v=0\land \overline{\Reg}_+(x,v)=\{0\} \land \Reg_-(x,v)=\emptyset\right) \lor$
\\
$\left(x_v=0\land \overline{\Reg}_-(x,v)=\{0\} \land \Reg_+(x,v)=\emptyset\right).$
\end{tabular}
\end{equation}
These two conditions are absorbed within the disjunction under the clauses \eqref{2prime} and \eqref{3prime}. This proves the claim of the theorem.
\end{proof}

The following result is a simple consequence of Definition \ref{d_SRG} and Theorem \ref{t_SRG} and it offers more insight into the dynamics of strong regulatory graphs.

\begin{lemma}
Let $G=(V,E)$ be a strong regulatory graph, $v\in V$ and $x$ a state of $G$. Then

\begin{enumerate}[(i)]
\item $f_G(x)_v=x_v$, if $x_v=1$ and $\Reg_-(x,v)=\emptyset$, or $x_v=-1$ and $\Reg_+(x,v)=\emptyset$, or $x_v=0$ and $1\not\in\Reg_+(x,v)\cup\Reg_-(x,v)$;

\item $f_G(x)_v=1$, if $1\in\Reg_+(x,v)$ and $\Reg_-(x,v)=\emptyset$;

\item $f_G(x)_v=-1$, if $1\in\Reg_-(x,v)$ and $\Reg_+(x,v)=\emptyset$;

\item $f_G(x)_v=0$, if $\Reg_+(x,v)\ne\emptyset$ and $\Reg_-(x,v)\ne\emptyset$.
\end{enumerate}
\end{lemma}

\section{Phenotype attractors}

We are interested in attractors that are defined through a fixed configuration (also called phenotype in \cite{10.1109/TCBB.2018.2879097}) on some (possibly not all) of their variables. In other words, we are interested in minimal cycles of the state transition systems, where some of the variables of the models are constant. This is similar to the concept of a target set in the partial controlability of complex networks \cite{ControlSciRep2017}. The notion of phenotype was introduced for Boolean networks in \cite{10.1109/TCBB.2018.2879097} and discussed in connection with the control of Boolean networks in \cite{10.3389/fams.2022.838546}.

\begin{definition}
Let $G=(V,E)$ be a strong regulatory graph. We define a \emph{trap set} of its state transition graph in the usual way as a set $S$ of vertices such that $f_G(S)\subseteq S$. A trap set is also called an invariant in dynamical system theory. An \emph{attractor} is defined as a (non-empty) minimal trap set under set inclusion.

A \emph{target} $T$ is a set of vertices $T\subseteq V$. A \emph{$T$-phenotype} is a function $\alpha_T:T\rightarrow \{-1,1\}$ that gives an active/inactive status assignment of the vertices in $T$. An \emph{$\alpha_T$-phenotype attractor} ${\cal A}_{\alpha_T}$ is an attractor of the state transition graph that has the phenotype $\alpha_T$ on $T$, i.e., for any $x\in {\cal A}_{\alpha_T}$, $x_t=\alpha_T(t)$, for all $t\in T$.
\end{definition}

The \emph{phenotype problem} is to decide whether for a given phenotype $\alpha_T$, there exists an $\alpha_T$-phenotype attractor. We prove that the problem can be solved in polynomial time. Even more, we give a simple characterization of phenotype attractors.

We say that $u$ is an \emph{inhibition-predecessor} of $v$ if $(u,v)\in E_-$. We say that $u$ is an \emph{inhibition-ancestor} of $v$ if there are vertices $w_1, w_2,\dots, w_n$, $n\geq 1$ with $(w_i,w_{i+1})\in E_-$ for all $1\leq i\leq n-1$, $w_1=u$ and $w_n=v$.

Similarly, we say that $u$ is an \emph{activation-predecessor} of $v$ if $(u,v)\in E_+$.
We say that $u$ is an \emph{activation-ancestor} of $v$ if there are vertices $w_1, w_2,\dots, w_n$, $n\geq 1$ with $(w_i,w_{i+1})\in E_+$ for all $1\leq i\leq n-1$, $w_1=u$ and $w_n=v$.

\begin{theorem}\label{t_attractor}
Let $G(V,E)$ be a strong regulatory graph, let $T\subseteq V$ be a target set and let  $\alpha_T:T\rightarrow\{-1,1\}$ be a phenotype on $T$. There exists an $\alpha_T$-attractor if and only if the following two conditions hold:
\begin{enumerate}[(i)]
    \item\label{target1} for all $v\!\in T$ with $\alpha_T(v)\!=1$ and for all $u\!\in T$ inhibition-predecessors of $v$, we have \mbox{$\alpha_T(u)\!=-1$;}
    \item\label{target2} for all $v\in T$ with $\alpha_T(v)\!=-1$ and for all $u\in T$ activation-ancestors of $v$, we have \mbox{$\alpha_T(u)\!=-1$.}
\end{enumerate}
\end{theorem}

\begin{proof}
Assume first that there is an $\alpha_T$-attractor ${\cal A}_{\alpha_T}$. This means that for all $x\in{\cal A}_{\alpha_T}$ and for all $t\in T$, $x_t=\alpha_T(t)$.
Consider $v\in T$. If $\alpha_T(v)=1$, then for all $x\in {\cal A}_{\alpha_T}$, $\Reg_-(x,v)=\emptyset$, i.e., its inhibition predecessors $v\in T$ have $\alpha_T(v)=-1$. If $\alpha_T(v)=-1$, then for all $x\in{\cal A}_{\alpha_T}$, $\Reg_+(x,v)=\emptyset$, i.e., its activation predecessors $v$ have $\alpha_T(v)=-1$. This argument can be iterated throughout all activation predecessors, yielding the second claim.

Consider now a phenotype $\alpha_T$ satisfying properties \eqref{target1} and \eqref{target2}. We can construct an $\alpha_T$-attractor ${\cal A}_{\alpha_T}$ in the following way. We visit the graph going against the edges, starting from the target set $T$, and marking them with $-1$ or $1$ as we visit them. Let $S$ be the set of nodes still to be explored. We start the exploration with $S=T$, whose marking is already set through $\alpha_T$. Take all vertices in $S$ marked with $1$, consider all their inhibition predecessors, mark them with $-1$ and add them to $S$. Because of property \eqref{target1}, none of them was marked with $1$ and so, this marking will not contradict any previously set marking. Take all vertices in $S$ marked with $-1$, consider all their activation predecessors and mark them with $-1$. Because of property \eqref{target2}, none of them was marked with $1$ and so, this marking will not contradict any previously set marking. Iterate through this step until set $S$ stops growing. To construct an $\alpha_T$-attractor, consider a state $x_1\in\{-1,0,1\}^{|V|}$ defined on $S$ through the markings set above, and taking an arbitrary choice from $\{-1,0,1\}$ for the vertices in $V\setminus S$. We then consider the transitions $x_1\rightarrow x_2\rightarrow x_3\rightarrow\cdots$. Obviously, there will eventually be a repetition of states $x_i=x_j$, $i<j$ in this sequence. The $\alpha_T$-attractor is $\{x_i, x_{i+1}, \dots, x_{j-1}\}$.
\end{proof}

The following is a reformulation of Theorem \ref{t_attractor}.
A vertex $t$ is an \emph{active target} if $t\in T$ and $\alpha_T(t)=1$.
Similarly, A vertex $t$ is an \emph{inactive target} if $t\in T$ and $\alpha_T(t)=-1$.
An \emph{activation path} is a path in a strong regulatory graph made of activation edges only.

\begin{theorem}
Let $T\subseteq V$ be a target set. For phenotype $\alpha_T:T\rightarrow\{-1,1\}$ there is an $\alpha_T$-attractor if and only if
\begin{enumerate}[(a)]
\itemsep=0.9pt
    \item there is no activation path from an active target to an inactive target and
    \item there is no activation path plus a final inhibition edge between two active targets.
\end{enumerate}
\end{theorem}

\section{Application to a regulatory cancer network}

We discuss in this section a strong regulatory graph model of some of the key elements of the RTK (receptor tyrosine kinase) signaling through the MAPK- (mitogen - activated protein kinase) and the PI3K/AKT- (lipid kinase phoshoinositide-3-kinase) pathways. Our goal is to demonstrate that the strong regulatory graphs are expressive enough to capture some interesting properties of a biological model. The MAPK signaling pathway communicates signals from outside the cell to the nucleus of the cell, it is involved in cell growth and proliferation, it is often mutated in cancer and a popular target of cancer treatments \cite{mapk_pathway, mapk_review, mapk_review2}. It interplays with the PI3K/AKT signaling pathway, a master regulator of the cell, whose activation contributes to the development of tumors and resistance to anticancer therapies \cite{pi3k_target, pi3k_target2}. We consider here a portion of the Boolean network model of \cite{GTZANUDO20181}. The model is illustrated in Figure \ref{fig-mapk-pi3k},
with the key regulators shown as vertices, with pointed arrows for the activation edges and blunt arrows for the inhibition edges.
The model includes FOXO$_3$, a protein known for its role in inducing cell death, often inhibited in tumors, and AKT, with a role in cell survival and often activated in tumors. Following \cite{GTZANUDO20181}, the activation of RTK depends on the presence of growth factors, not included in the model, and so its state is set to be constant $RTK=-1$. We consider the $(-1,1)$ state of (FOXO$_3$, AKT) as an indicator of uncontrolled proliferation, that of $(1,-1)$ an indicator of non-proliferation, and that of $(-1,-1)$ an indicator of  moderate proliferation. (Because of the inactivation edge from AKT to FOXO$_3$, an $(1,1)$-attractor is not possible.) These are of course over-simplifications of a much more complex interplay of interactions.
We show that the strong regulatory graph model has the same variety of outcomes as the Boolean network of \cite{zanudo2011boolean}, while adopting the strong update rule proposed in Definition~\ref{d_SRG}.

\begin{figure}[!ht]
\vspace*{3mm}
\begin{center}
\includegraphics[width=5.56cm]{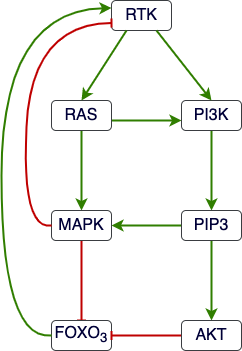}
\end{center}\vspace*{-4mm}
\caption{The MAPK-PI3K/AKT signaling pathways.}\label{fig-mapk-pi3k}
\end{figure}

In our discussion of the dynamics of the model, we write the variables in the order (RTK, RAS, PI3K, MAPK, PIP3, FOXO$_3$, AKT).
The model has several phenotype attractors associated with all three possibilities of the (FOXO$_3$, AKT) phenotype described above. Indeed, here are examples of $(-1,-1)$-, $(1,-1)$-, and $(-1,1)$-attractors:
\begin{itemize}
\itemsep=0.65pt
    \item for $x_1=(-1,-1,-1, 1,-1, {\bf -1}, {\bf -1})$, $x_1\rightarrow x_1$;
    \item for $x_2=(-1,-1,-1,-1,-1, {\bf 1}, {\bf -1})$, $x_2\rightarrow x_2$;
    \item for $x_3=(-1,1,1,1,1,{\bf -1}, {\bf 1})$, $x_3\rightarrow x_3$.
\end{itemize}

\begin{figure}[!b]
\begin{center}
\adjustbox{
clip=true, trim=0cm 0cm 3cm 0cm 
}
{
\includegraphics[width=1.1\textwidth]{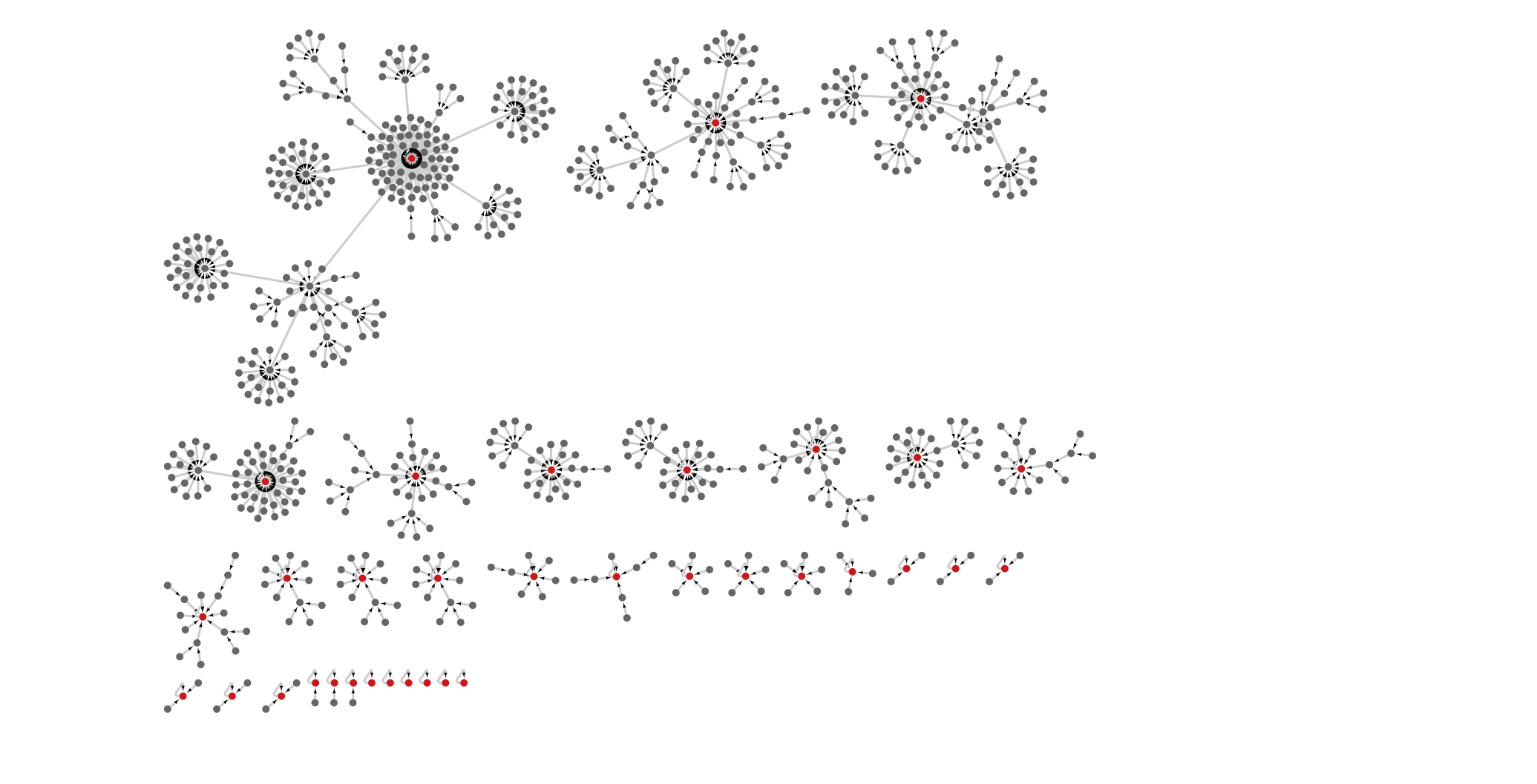}
}
\end{center}\vspace*{-12mm}
\caption{The basins of attractors of the state transition graph of the MAPK-PI3K/AKT model.}\label{fig-mapk-pi3k-attractors}
\end{figure}

A mutation on PI3K, making it constant active, gives proliferation as the only stable set, consistent with the biological observations in \cite{pi3k_target} and with the behavior of the Boolean network model of \cite{zanudo2011boolean}: for $x_4=(-1,-1, {\bf 1},1,1, {\bf -1}, {\bf 1})$ we have $x_4\rightarrow x_4$.
Also consistent with the Boolean network model of \cite{zanudo2011boolean} is the observation that setting PI3K to inactive (corresponding to the idea of targeting it with drug inhibitors), allows all proliferation outcomes to be possible:
\begin{itemize}
\itemsep=0.95pt
    \item $(-1,-1,{\bf -1},-1,1, {\bf 1}, {\bf -1})\rightarrow (-1,-1,{\bf -1},1,1, {\bf 1}, {\bf 1})\rightarrow (-1,-1,{\bf -1},1,1,{\bf -1}, {\bf 1})\rightarrow(-1,$ $-1, {\bf -1},1,1,{\bf -1}, {\bf 1})$, with the final vertex being a singleton attractor;
    \item $(-1,-1,{\bf -1},-1,-1, {\bf 1}, {\bf -1})$ is a singleton attractor;
    \item $(-1,-1, {\bf -1},1,-1, {\bf 1}, {\bf -1} )\rightarrow (-1,-1,{\bf -1},1,-1,{\bf -1}, {\bf -1})$, with the final vertex being a singleton attractor.
\end{itemize}


The state transition graph of the model (with $RTK=-1$) has 729 nodes and 35 attractors. The graph is in Figure \ref{fig-mapk-pi3k-attractors}, with the basins of attractors made visible (but not the labels).

\section{Discussion}
\label{sec:disc}

Our main motivation in this article was to offer a solution to the gap in the current modeling of large network models. Current network models span between structural modeling approaches (well-supported by interaction data, scalable, but not dynamical) and detailed logical models (dynamical, but problematic to scale up because the data is typically not detailed enough to specify the update functions of each variable).
We introduced the concept of strong regulatory graph, a form of logical network model where the active/inactive status of a vertex is set only if its regulators agree in their influences; otherwise, the status of the vertex is set to ambiguous. Ambiguity here means that the node may potentially be active or inactive, depending on the state configuration of its predecessors, a level of detail that is left out of logical models. The ambiguity may have a cascading influence over the successors of the node, but it may also cancel out through the clarifications brought by other active vertices in the network. We discussed an update rule that defines the interplay between active, inactive, and ambiguous regulators. We also discussed the phenotype attractor problem and showed that it is easy to decide whether an attractor of a given phenotype exists in a strong regulatory graph.

\medskip
It is straightforward to see that for any strong regulatory graph there is a Boolean network that can simulate it. The state transition rule is indeed based on simple logical tests, and so easily implementable through Boolean functions. Also, the range of values $\{-1,0,1\}$ can easily be implemented using two Boolean variables for each vertex of the strong regulatory graph. This observation connects the strong regulatory graphs to the rich literature on the control of Boolean networks.

\medskip
We discussed only the synchronous version of the strong regulatory graphs, where all vertices are updated simultaneously. Conceptually, the asynchronous version is similar, and it can be discussed through both deterministic, as well as non-deterministic state-transition systems, whose properties should be interesting to characterize in a further study. Also left for another study is the study of controllability of strong regulatory graphs. Its targeted version can be easily defined in terms of reaching a phenotype attractor from any state with some minimal interventions in the state transition system.

\subsection*{Acknowledgments}
We are very grateful to the referees for their very detailed and very helpful suggestions, that supported us in improving the readability of the paper. This work received support from MRID, project no 842027778, contract no 760096 and the Core Program within the National Research, Development and Innovation Plan 2022-2027, carried out with the support of MRID, project no. 23020101(SIA-PRO), contract no 7N/2022.


\begin{thebibliography}{10}
\providecommand{\url}[1]{\texttt{#1}}
\providecommand{\urlprefix}{URL }
\expandafter\ifx\csname urlstyle\endcsname\relax
  \providecommand{\doi}[1]{doi:\discretionary{}{}{}#1}\else
  \providecommand{\doi}{doi:\discretionary{}{}{}\begingroup
  \urlstyle{rm}\Url}\fi
\providecommand{\eprint}[2][]{\url{#2}}

\bibitem{Bloomingdale2018}
Bloomingdale P, Nguyen VA, Niu J, Mager DE.
\newblock Boolean network modeling in systems pharmacology.
\newblock \emph{Journal of Pharmacokinetics and Pharmacodynamics}, 2018.
\newblock \textbf{45}(1):159--180.
\newblock \doi{https://doi.org/10.1007/s10928-017-9567-4}.
\newblock \urlprefix\url{https://doi.org/10.1007/s10928-017-9567-4}.

\bibitem{CACACE2020205}
Cacace E, Collombet S, Thieffry D.
\newblock Chapter Seven - Logical modeling of cell fate
  specification-Application to T cell commitment.
 In: Peter IS (ed.), Gene Regulatory Networks, volume 139 of
  \emph{Current Topics in Developmental Biology}, pp. 205--238. Academic Press,
  2020.  \doi{https://doi.org/10.1016/} {bs.ctdb.2020.02.008}.
\newblock \urlprefix\url{https://doi.org/10.1016/bs.ctdb.2020.02.008}.

\bibitem{10.7554/eLife.72626}
Montagud A, B\'eal J, Tobalina L, Traynard P, Subramanian V, Szalai B, Alf\"oldi
  R, Pusk\'as L, Valencia A, Barillot E, Saez-Rodriguez J, Calzone L.
\newblock Patient-specific Boolean models of signalling networks guide
  personalised treatments.
\newblock \emph{eLife}, 2022.
\newblock \textbf{11}:e72626.
\newblock \doi{https://doi.org/10.7554/eLife.72626}.
\newblock \urlprefix\url{https://doi.org/10.7554/eLife.72626}.

\bibitem{10.1371/journal.pcbi.1006402}
Sizek H, Hamel A, Deritei D, Campbell S, Ravasz~Regan E.
\newblock Boolean model of growth signaling, cell cycle and apoptosis predicts
  the molecular mechanism of aberrant cell cycle progression driven by
  hyperactive PI3K.
\newblock \emph{PLOS Computational Biology}, 2019.
\newblock \textbf{15}(3):1--43.
\newblock \doi{https://doi.org/10.1371/journal.pcbi.1006402}.
\newblock \urlprefix\url{https://doi.org/10.1371/journal.pcbi.1006402}.

\bibitem{10.3389/fphys.2020.590479}
Checcoli A, Pol JG, Naldi A, Noel V, Barillot E, Kroemer G, Thieffry D, Calzone
  L, Stoll G.
\newblock Dynamical Boolean Modeling of Immunogenic Cell Death.
\newblock \emph{Frontiers in Physiology}, 2020.
\newblock \textbf{11}.
\newblock \doi{https://doi.org/10.3389/fphys.2020.590479}.
\newblock \urlprefix\url{https://doi.org/10.3389/fphys.2020.590479}.

\bibitem{Wang_2012}
Wang RS, Saadatpour A, Albert R.
\newblock Boolean modeling in systems biology: an overview of methodology and
  applications.
\newblock \emph{Physical Biology}, 2012.
\newblock \textbf{9}(5):055001.
\newblock \doi{https://doi.org/10.1088/1478-3975/9/5/055001}.
\newblock \urlprefix\url{https://doi.org/10.1088/1478-3975/9/5/055001}.

\bibitem{10.1093/nar/gkv1070}
Kanehisa M, Sato Y, Kawashima M, Furumichi M, Tanabe M.
\newblock {KEGG as a reference resource for gene and protein annotation}.
\newblock \emph{Nucleic Acids Research}, 2015.
\newblock \textbf{44}(D1):D457--D462.
\newblock \doi{https://doi.org/10.1093/nar/gkv1070}.
  \urlprefix\url{https://doi.org/10.1093/nar/gkv1070}.

\bibitem{article-omnipath}
T\"urei D, Korcsmáros T, Saez-Rodriguez J.
\newblock {OmniPath: guidelines and gateway for literature-curated signaling
  pathway resources}.
\newblock \emph{Nature Methods}, 2016.
\newblock \textbf{13}:966--967.
\newblock \doi{https://doi.org/10.1038/nmeth.4077}.

\bibitem{article-innate}
Breuer K, Foroushani AK, Laird MR, Chen C, Sribnaia A, Lo R, Winsor GL, Hancock
  REW, Brinkman FSL, Lynn DJ.
\newblock {InnateDB: systems biology of innate immunity and beyond--recent
  updates and continuing curation}.
\newblock \emph{Nucleic Acids Research}, 2013.
\newblock \textbf{41}(Database issue):D1228--D1233.
\newblock \doi{https://doi.org/10.1093/nar/gks1147}.

\bibitem{article-signor}
Licata L, Lo Surdo P, Iannuccelli M, Palma A, Micarelli E, Perfetto L, Peluso
  D, Calderone A, Castagnoli L, Cesareni G.
\newblock {SIGNOR 2.0, the SIGnaling Network Open Resource 2.0: 2019 update}.
\newblock \emph{Nucleic Acids Research}, 2020.
\newblock \textbf{48}(D1):D504--D510.
\newblock \doi{https://doi.org/10.1093/nar/gkz949}.

\bibitem{article-drugbank}
Wishart DS, Feunang YD, Guo AC, Lo EJ, Marcu A, Grant JR, Sajed T, Johnson D,
  Li C, Sayeeda Z, Assempour N, Iynkkaran I, Liu Y, Maciejewski A, Gale N,
  Wilson A, Chin L, Cummings R, Le D, Pon A, Knox C, Wilson M.
\newblock {DrugBank 5.0: a major update to the DrugBank database for 2018}.
\newblock \emph{Nucleic Acids Research}, 2017.
\newblock \textbf{46}(D1):D1074--D1082.
\newblock \doi{https://doi.org/10.1093/nar/gkx1037}.

\bibitem{10.1093/bib/bbab490}
Siminea N, Popescu V, Sanchez~Martin JA, Florea D, Gavril G, Gheorghe AM,
  I\k{t}cu\k{s} C, Kanhaiya K, Pacioglu O, Popa LI, Trandafir R, Tusa MI, Sidoroff M,
  Păun M, Czeizler E, Păun A, Petre I.
\newblock {Network analytics for drug repurposing in COVID-19}.
\newblock \emph{Briefings in Bioinformatics}, 2021.
\newblock \textbf{23}(1).
 \doi{https://doi.org/10.1093/bib/bbab490}.
 Bbab490,  
  \urlprefix\url{https://doi.org/10.1093/bib/bbab490}.

\bibitem{NetControlGenAlg}
Popescu VB, Kanhaiya K, N{\u a}stac DI, Czeizler E, Petre I.
\newblock Network controllability solutions for computational drug repurposing
  using genetic algorithms.
\newblock \emph{Scientific Reports}, 2022.
\newblock \textbf{12}(1):1437.
\newblock \doi{https://doi.org/10.1038/s41598-022-05335-3}.
\newblock \urlprefix\url{https://doi.org/10.1038/s41598-022-05335-3}.

\bibitem{GTZANUDO20181}
{Zañudo GT} J, Steinway SN, Albert R.
\newblock Discrete dynamic network modeling of oncogenic signaling: Mechanistic
  insights for personalized treatment of cancer.
\newblock \emph{Current Opinion in Systems Biology}, 2018.
\newblock \textbf{9}:1--10.
 \doi{https://doi.org/10.1016/j.coisb.2018.02.002}.
 Mathematic modelling,   \urlprefix\url{https://doi.org/10.1016/j.coisb.2018.02.002}.

\bibitem{zanudo2011boolean}
Zanudo JG, Aldana M, Mart{\'\i}nez-Mekler G.
\newblock Boolean threshold networks: Virtues and limitations for biological
  modeling.
\newblock In: Information Processing and Biological Systems, pp. 113--151.
  Springer, 2011.
\newblock \doi{{https://doi.org/10.1007/978-3-642-19621-8_6}}.
\newblock \urlprefix\url{{https://doi.org/10.1007/978-3-642-19621-8_6}}.

\bibitem{naldi_decision_2007}
Naldi A, Thieffry D, Chaouiya C.
\newblock Decision {Diagrams} for the {Representation} and {Analysis} of
  {Logical} {Models} of {Genetic} {Networks}.
\newblock In: Calder M, Gilmore S (eds.), Computational {Methods} in {Systems}
  {Biology}. Springer Berlin Heidelberg, Berlin, Heidelberg.
\newblock ISBN 978-3-540-75140-3, 2007 pp. 233--247.
\newblock \doi{https://doi.org/10.1007/978-3-540-75140-3_16}.
 \newblock \urlprefix\url{https://doi.org/10.1007/978-3-540-75140-3_16}.

\bibitem{StructuralTargetControl2018}
Czeizler E, Gratie C, Chiu WK, Kanhaiya K, Petre I.
\newblock Structural target controlability of linear networks.
\newblock \emph{IEEE/ACM Transactions on Computational Biology and
  Bioinformatics}, 2018.
\newblock \textbf{15}(4):1217 -- 1228.
\newblock \doi{https://doi.org/10.1109/TCBB.2018.2797271}.
 \urlprefix\url{https://doi.org/10.1109/TCBB.2018.2797271}.

\bibitem{ControlRSystems2020}
Ivanov S, Petre I.
\newblock Controllability of reaction systems.
\newblock \emph{Journal of Membrane Computing}, 2020.
\newblock \textbf{7}:290 –-- 302.
\newblock \doi{https://doi.org/10.1007/s41965-020-00055-x}.
\newblock \urlprefix\url{https://doi.org/10.1007/s41965-020-00055-x}.

\bibitem{thomas1991regulatory}
Thomas R.
\newblock Regulatory networks seen as asynchronous automata: a logical
  description.
\newblock \emph{Journal of Theoretical Biology}, 1991.
\newblock \textbf{153}(1):1--23.
\newblock \doi{https://doi.org/10.1016/S0022-5193(05)80350-9}.
\newblock \urlprefix\url{https://doi.org/10.1016/S0022-5193(05)80350-9}.

\bibitem{thomas1995dynamical}
Thomas R, Thieffry D, Kaufman M.
\newblock Dynamical behaviour of biological regulatory networks---I. Biological
  role of feedback loops and practical use of the concept of the
  loop-characteristic state.
\newblock \emph{Bulletin of Mathematical Biology}, 1995.
\newblock \textbf{57}(2):247--276.
\newblock \doi{https://doi.org/10.1007/BF02460618}.
 \newblock \urlprefix\url{https://doi.org/10.1007/BF02460618}.

\bibitem{GIN-sim}
Chaouiya C, Remy E, Moss{\'e} B, Thieffry D.
\newblock Qualitative Analysis of Regulatory Graphs: A Computational Tool Based
  on a Discrete Formal Framework.
\newblock In: Benvenuti L, De~Santis A, Farina L (eds.), Positive Systems.
  Springer Berlin Heidelberg, Berlin, Heidelberg.
\newblock ISBN 978-3-540-44928-7, 2003 pp. 119--126.
\newblock \doi{http://doi.org/10.1007/978-3-540-44928-7_17}.
 \newblock \urlprefix\url{http://doi.org/10.1007/978-3-540-44928-7_17}.

\bibitem{10.1109/TCBB.2018.2879097}
Klarner H, Heinitz F, Nee S, Siebert H.
\newblock Basins of Attraction, Commitment Sets, and Phenotypes of Boolean Networks.
\newblock \emph{IEEE/ACM Trans. Comput. Biol. Bioinformatics}, 2020.
\newblock \textbf{17}(4):1115–1124.
\newblock \doi{https://doi.org/10.1109/TCBB.2018.2879097}.
 \newblock \urlprefix\url{https://doi.org/10.1109/TCBB.2018.2879097}.

\bibitem{ControlSciRep2017}
Kanhaiya K, Czeizler E, Gratie C, Petre I.
 Controlling Directed Protein Interaction Networks in Cancer.
 \emph{Scientific Reports}, 2017.
 \textbf{7}(1):10327.
 \doi{https://10.1038/s41598-017-10491-y}.
 \urlprefix\url{https://10.1038/s41598-017-10491-y}.

\bibitem{10.3389/fams.2022.838546}
Cifuentes-Fontanals L, Tonello E, Siebert H.
\newblock Control in Boolean Networks With Model Checking.
\newblock \emph{Frontiers in Applied Mathematics and Statistics}, 2022.
\newblock \textbf{8}.
\newblock \doi{https://doi.org/10.3389/fams.2022.838546}.
 \newblock \urlprefix\url{https://doi.org/10.3389/fams.2022.838546}.

\bibitem{mapk_pathway}
Yaeger R, Corcoran RB.
\newblock {Targeting Alterations in the RAF–MEK Pathway}.
\newblock \emph{Cancer Discovery}, 2019.
\newblock \textbf{9}(3):329--341.
\newblock \doi{https://doi.org/10.1158/2159-8290.CD-18-1321}.
 \urlprefix\url{https://doi.org/10.1158/2159-8290.CD-18-1321}.

\bibitem{mapk_review}
Burotto M, Chiou VL, Lee JM, Kohn EC.
\newblock The MAPK pathway across different malignancies: A new perspective.
\newblock \emph{Cancer}, 2014.
\newblock \textbf{120}(22):3446--3456.
\newblock \doi{https://doi.org/10.1002/cncr.28864}.
 \urlprefix\url{https://doi.org/10.1002/cncr.28864}.

\bibitem{mapk_review2}
Dhillon A, Hagan S, Rath O, Kolch W.
\newblock MAP kinase signalling pathways in cancer.
\newblock \emph{Oncogene}, 2007.
\newblock \textbf{26}:3279–3290.
\newblock \doi{https://doi.org/10.1038/sj.onc.1210421}.

\bibitem{pi3k_target}
Yang J, Nie J, Ma X, Wei Y, Peng Y, Wei X.
\newblock Targeting PI3K in cancer: mechanisms and advances in clinical trials.
\newblock \emph{Mol Cancer}, 2019.
\newblock \textbf{18}(26).
\newblock \doi{https://doi.org/10.1186/s12943-019-0954-x}.

\bibitem{pi3k_target2}
He Y, Sun MM, Zhang GG, Yang J, Chen KS, Xu WW, Li B.
\newblock TTargeting PI3K/Akt signal transduction for cancer therapy.
\newblock \emph{Sig Transduct Target Ther}, 2021.
\newblock \textbf{6}(425).
\newblock \doi{https://doi.org/10.1038/s41392-021-00828-5}.
\end{thebibliography}
\end{document}